\documentclass{article}
\usepackage{planargraph}
\usepackage{theoremsetc}
\usepackage{graphicx}
\usepackage{amsmath,amssymb}
\usepackage{algorithm}
\usepackage[noend]{algorithmic}

\newenvironment{proof}
{\paragraph{Proof:}}
{QED\\}

\newcommand{\inflow}{\textrm{inflow}}
\newcommand{\outflow}{\textrm{outflow}}
\begin{document}

\title{Multiple-source single-sink maximum flow in directed planar
             graphs in $O(n^{1.5} \log n)$ time}

\author{Philip N. Klein and Shay Mozes\\Brown University }

\maketitle

\begin{abstract}
  We give an $O(n^{1.5} \log n)$ algorithm that, given a directed
  planar graph with arc capacities, a set of source nodes and a single
  sink node, finds a maximum flow from the sources to the sink . This
 is the first subquadratic-time strongly polynomial algorithm for the
 problem. 
\end{abstract}

\nocite{Weihe97,Hassin81,Frederickson87,IS79,JV82,FF56,FR01,HKRS97,HJ85,Reif83,HKRS97,EK72,Dinitz70}

\section{Introduction}

The study of maximum flow in planar graphs has a long history.  In
1956, Ford and Fulkerson introduced the max $st$-flow problem, gave a
generic augmenting-path algorithm, and also gave a particular
augmenting-path algorithm for the case of a planar graph where $s$ and
$t$ were on the same face (that face is traditionally designated to be
the infinite face).  Researchers have since published many
algorithmic results proving running-time bounds on max $st$-flow for
(a) planar graphs where $s$ and $t$ are on the same face, (b)
undirected planar graphs where $s$ and $t$ are arbitrary, and (c)
directed planar graphs where $s$ and $t$ are arbitrary.  The best
bounds known are (a) $O(n)$~\cite{HKRS97}, (b) $O(n \log
n)$~\cite{Frederickson87}, and (c) $O(n \log n)$~\cite{BorradaileK09}.

Schrijver~\cite{Schrijver02} has written about the history of this problem.
Ford and Fulkerson, who worked at RAND, were apparently motivated by
classified work of Harris and Ross on interdiction
of the Soviet railroad system.  (Of course, Harris and Ross were
interested in the min cut, not the max flow, as seems to be true for
most applications.)  This article was downgraded to unclassified in
1999.  It contains a diagram of a network that models the Soviet
railroad system indicating ``ORIGINS'' (sources) and what is
apparently a sink (marked ``EG'').

\begin{figure}
\centerline{\includegraphics[width=5in]{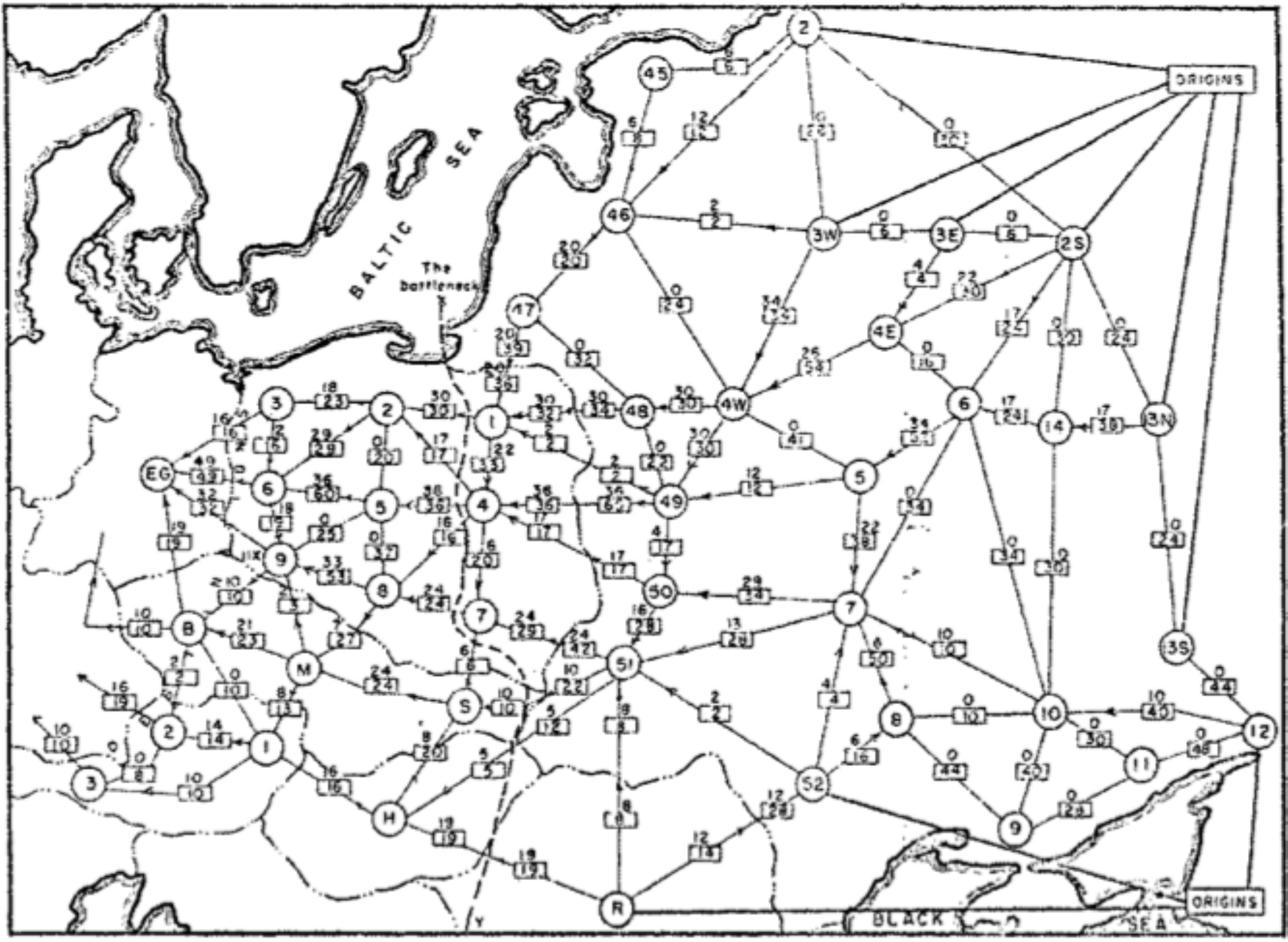}}
\caption{The soviet rail network}
\end{figure}

In max-flow applied to general
graphs, multiple sources presents no problem: one can reduce the
problem to the single-source case by introducing an artificial source
and connecting it to all the sources.  However, as Miller and
Naor~\cite{MN95} pointed out, this reduction violates planarity unless
all the sources are on the same face to begin with.  Miller and Naor
raise the question of computing a maximum flow in a planar graph with
multiple sources and multiple sinks.  Even when there is one sink,
until now the best known algorithm for computing multiple-source
max-flow in a planar graph is to use the reduction in conjunction with
a max-flow algorithm for general graphs.  That is, no
planarity-exploiting algorithm was known for the problem.

There are workarounds.  For example, in the Soviet rail network, there
are two faces that together include all the sources, so solving the
instance can be reduced to solving two single-source max flows in a
planar graph.  However, a more realistic motivation comes from
selecting multiple nonoverlapping regions in a planar structure.  

Consider, for example, the following image-segmentation problem.  A
grid is given in which each vertex represents a pixel, and edges
connect orthogonally adjacent pixels.  Each edge is assigned a cost
such that the edge between two similar pixels has higher cost than
that between two very different pixels.  In addition, each pixel is
assigned a weight.  High weight reflects a high likelihood that the
pixel belongs to the foreground; a low-weight pixel is more likely to
belong to the background.

The goal is to find a partition of the pixels into foreground and
background to minimize the sum 
\begin{eqnarray*}
\lefteqn{\mbox{weight of  {\it background} pixels}}\\
& + & \mbox{cost of
  edges between {\it foreground} pixels and {\em background} pixels}
\end{eqnarray*}
subject to the constraints that, for each component $K$ of foreground
pixels, the boundary of $K$ forms a closed curve in the planar dual
that surrounds all of $K$ (essentially that the component is simply connected).

This problem can be reduced to multiple-source, single-sink max-flow
in a planar graph (in fact, essentially the grid).  For each pixel
vertex $v$, a new vertex $v'$, designated a source, is introduced and
connected only to $v$.  Then the sink is connected to the pixels at
the outer boundary of the grid.

\paragraph{New result}

We prove the following:
\begin{theorem}\label{thm:main}
There is an $O(n^{1.5} \log n)$ algorithm to compute multiple-source,
single-sink max flow in an $n$-node directed planar graph.
\end{theorem}

Before our work, the best strongly polynomial bound for the problem is
$O(n^2 \log n)$, which comes from the reduction to general graphs and
then use of an algorithm such as that of Goldberg and
Tarjan~\cite{GT88}.  For integer capacities less than $U$, one could
instead use the algorithm of Goldberg and Rao~\cite{GR98}, which leads
to a running time of $O(n^{1.5} \log n \log U)$, or the
planarity-exploiting min-cost flow algorithm of~\cite{ImaiIwano},
which gives a bound of $O(n^{1.595} \log U$) that depends on fast
matrix multiplication and interior-point methods.  However, even if
one assumes integer capacities and $U=\Theta(n)$, our planarity-exploiting
algorithm is asymptotically faster.

We have learned (personal communication) that Borradaile and Wulff-Nilsen
have independently proved the same theorem.

\subsection{Organization}
The structure of the paper is as follows. In Section~\ref{sec:prel} we
give some definitions and general technical background. In Section~\ref{sec:algo},
we give the main algorithm. Finally, in Section~\ref{sec:preflow} we
describe how to efficiently convert a feasible preflow to a feasible flow.

\section{Preliminaries}\label{sec:prel}

\subsection{Embedded Planar Graphs}

A {\em planar embedding} of a graph assigns each node to a distinct
point onto the sphere, and assigns each edge to a simple arc between the points
corresponding to its endpoints, with the property
that no arc-arc or arc-point intersections occur except for those
corresponding to edge-node incidence in the graph.  A graph is planar
if it has a planar embedding. 

Assume the graph is connected, and
consider the set
of points on the sphere that are not assigned to any node or edge;
each connected component of this set is a {\em face} of the
embedding.

It is convenient to designate one face as the {\em infinite face} (by
analogy to embeddings on the plane).  With respect to a choice of the
infinite face, we say a Jordan curve strictly encloses an edge or node
if the Jordan curve separates the edge or node from the infinite
face.  Similarly, for a subgraph, the
choice of infinite face $f_\infty$ for the whole graph induces a 
choice of infinite face for each connected component of the subgraph,
namely that face of the connected component that contains $f_\infty$.

In implementations, an embedding onto the sphere can be represented
combinatorially, using a {\em rotation system}.

\subsection{Flow}
Let $G$ be a directed graph with arc set $A$, node set $V$ and sink $t$.
For notational simplicity, we assume here and henceforth that $G$ has
no parallel edges and no self-loops.

We associate with each arc $a$ two darts $d$ and $d'$, 
one in the direction of $a$ and the other in the opposite direction.
We say that those two darts are reverses of each other, and write $d = \rev(d')$.

A {\em flow assignment} $f(\cdot)$ is a real-valued function
on darts that satisfies {\em antisymmetry}:
\begin{equation} \label{eq:antisymmetry}
f(\rev(d))= -f(d)
\end{equation}

A {\em capacity assignment} $c(\cdot)$ is a function from darts to
real numbers.  A flow assignment $f(\cdot)$ {\em respects capacities}
if, for every dart $d$, $f(d) \leq c(d)$.  Note that, by antisymmetry,
$f(d) \leq c(d)$ implies $f(\rev(d)) \geq -c(\rev(d))$.  Thus a negative
capacity on a dart acts as a lower bound on the flow on the reverse
dart.  In this paper, we assume all capacities are nonnegative, and
therefore the all-zeroes flow respects the capacities.

For a given flow assignment $f(\cdot)$, the {\em net inflow} (or just {\em inflow}) node $v$ is $\inflow_f(v) = \sum_{a \in A : \head(a) = v}
f(a) - \sum_{a \in A : \tail(a) = v} $. The {\em outflow} of $v$ is
$\outflow_f(v) = - \inflow_f(v)$.  The {\em value} of $f(\cdot)$ is the
  inflow at the sink, $\inflow_f(t)$.

A flow assignment $f(\cdot)$ is said to \emph{obey conservation} if for every node $v$ other than $t$,
$\outflow_f(v) \geq 0$.

A {\em supply assignment} $\sigma(\cdot)$ is a function from the
non-sink nodes to $\mathbb{R} \cup \set{\infty}$. For any node $v$,
$\sigma(v)$ specifies the amount of flow that can originate at $v$. 
A flow assignment $f(\cdot)$ is said to \emph{respect the supplies $\sigma(\cdot)$}
if, for every node $v$ other than the sink $t$, $\outflow_f(v) \leq
\sigma(v)$.  In this paper, we assume all supply values are nonnegative.

A flow assignment is a \emph{feasible preflow} if it respects both capacities
and supplies.  
A feasible preflow is called a \emph{feasible flow} if in addition it obeys conservation. 
In this paper, we give an algorithm to find a maximum (feasible)
preflow, and then an algorithm to convert that preflow to a maximum
(feasible) flow.

The \emph{residual graph} of $G$ with respect to a flow assignment $f(\cdot)$ is the graph $G_f$ with
the same arc-set, node-set and sink, and with capacity assignment $c_f(\cdot)$
and supply assignment $\sigma_f(\cdot)$ defined as follows:
\begin{itemize}
\item  For every dart $d$, $c_f(d) =
c(d) - f(d)$.
\item  For every node $v$, $\sigma_f(v) = \sigma(v) -
\outflow_f(v)$.
\end{itemize}

\subsubsection*{Single-source limited max flow}

For a particular node $s$, a {\em limited max $st$-flow} is a flow
assignment $f(\cdot)$ of maximum value that obeys capacities and for
which $\inflow_f(v)=0$ for every node except $s$ and $t$ and such that
$\outflow_f(s) \leq \sigma(s)$.  An algorithm for ordinary max
$st$-flow can be used to compute limited max $st$-flow by introducing
an artificial node $s'$ and an arc $s' s$ of capacity $\sigma(s)$, and
running the algorithm on the transformed graph.
This transformation preserves planarity.
Since there is an $O(n \log n)$ algorithm for max $st$-flow in a
planar directed graph~\cite{BorradaileK09}, we assume a subroutine for
limited max $st$-flow.

\subsection{Jordan Separators for Embedded Planar Graphs}\label{sec:separator}

For an $n$-node planar embedded simple graph $G$, we define a {\em Jordan separator} to be a Jordan curve
$S$ such that, for any arc $a$ of $G$, the set of points in the sphere
corresponding to $a$ either (i) does not intersect $S$ or (ii) coincides with a
subcurve of $S$.  We require in addition that, if the two endpoints of
$a$ are consecutive nodes on $S$, then (ii) must hold.
The {\em boundary nodes} of $S$ are the nodes $S$ goes through.

We say a Jordan separator is {\em balanced} if
at most $2n/3$ nodes are strictly enclosed by the curve and
at most $2n/3$ nodes are not enclosed.  

Miller~\cite{Miller86} gave a linear-time algorithm that, given a
triangulated two-connected $n$-node planar embedded graph, finds
a simple cycle in the graph, consisting of at most $2\sqrt{2}\sqrt{n}$ nodes, such
that at most $2n/3$ nodes are strictly enclosed by the cycle, and at
most $2n/3$ nodes are not enclosed.

To find a balanced Jordan separator in a graph that is not necessarily triangulated or
two-connected, add artificial edges to
triangulate the graph and make it two-connected. 
Now apply Miller's algorithm to find a simple cycle separator with the desired property.
Viewed as a curve in the sphere, the resulting separator $S$ satisfies
the requirements of a balanced Jordan separator, and it has at most
$2\sqrt{2} \sqrt{n}$ boundary nodes.

\section{The Algorithm}\label{sec:algo}

The main algorithm finds a maximum preflow in the
following, slightly more general, setting.
\begin{itemize}
\item {\em Input:}
\begin{itemize}
\item A directed planar embedded graph $G$,
\item a sink node $t$,
\item a nonnegative capacity
assignment $c(\cdot)$, and
\item a nonnegative supply assignment $\sigma(\cdot)$.
\end{itemize}
\item {\em Output:} A feasible preflow $G$ of maximum value.
\end{itemize}

We present the main algorithm as a recursive procedure with calls to a
single-source limited-max-flow subroutine.  We omit discussion of the
base case of the recursion (the case where the graph size is smaller
than a certain constant.)  Each of the recursive
calls operates on a subgraph of the original input graph.  We assume
one global flow assignment $f(\cdot)$ for the original input graph,
one global capacity assignment $c(\cdot)$, and one global supply
assignment $\sigma(\cdot)$.  Whenever the single-source
limited-max-flow subroutine is called, it takes as part of its input 
\begin{itemize}
\item the current residual capacity function $c_f(\cdot)$ and 
\item the current
residual supply function $\sigma_f(\cdot)$.
\end{itemize} 
 It computes a limited
max flow assignment $\widehat f(\cdot)$, and then updates the global flow
assignment $f(\cdot)$ by $f(d) := f(d) + \widehat f(d)$ for every dart
in the subgraph.

In the pseudocode, we do not explicitly mention $f(\cdot)$, $c(\cdot)$
$\sigma(\cdot)$, $c_f(\cdot)$, or $\sigma_f(\cdot)$.

The pseudocode for the algorithm is given below. We assume that the
sink is on the boundary of the face designated the infinite face.

\begin{algorithm}\caption{MultipleSourceMaxPreFlow(graph $G_0$, sink
    $t$) \label{alg:multisource}}
\begin{algorithmic}[1]
\STATE triangulate $G_0$ with zero-capacity edges. 
\STATE $i:=0$
\WHILE{$G_i$ consists of more than $N_0$ nodes}
\STATE $i:=i+1$
\STATE find a Jordan separator $S_i$ in $G_{i-1}$\label{sep}
\STATE let $H_i$ be the subgraph of $G_{i-1}$ enclosed by $S_i$
\STATE let $C_i$ be the external face of $H_i$.
\STATE let $B_i$ be the set of cycles $\{C_j : C_j \textrm{
  is contained in } H_i \}$
\FOR {$C$ in $B_i$}\label{first}
  \STATE designate one of the nodes of $C$ as an artificial sink $t'$
  and add artificial infinite-capacity edges parallel to $C$
  \STATE  MultipleSourceMaxPreFlow($H_i,t'$)\label{push1}
  \STATE remove the infinite-capacity artificial edges
\ENDFOR
\STATE let $G_i$ be the subgraph of $G_{i-1}$ that is not strictly
enclosed by $S_i$\label{Gi}
\ENDWHILE

\FOR {$C$ in $\{C_j\}$}\label{second}
  \FOR {every node $v$ of $C$}\label{inner}
     \STATE limited max-flow from $v$ to $t$ in
     $G$\label{limited}
 \ENDFOR
\ENDFOR
\end{algorithmic}
\end{algorithm}

The algorithm proceeds in iterations as long as the current graph,
$G_i$, 
consists of more than $N_0$ nodes, for some constant $N_0$ to be
specified later. For graphs of constant size, output the solution in
constant time.  At iteration $i$ it finds a
small Jordan separator $S_i$ in $G_{i-1}$ as described in
Section~\ref{sec:separator}. Let $H_i$ be the subgraph of $G_{i-1}$ enclosed
by $S_i$. Intuitively, one would like to think of $S_i$ as the external
face of $H_i$. However, $S_i$ might cross some earlier
$S_j$, so $S_i$ does not entirely coincide with edges of $H_i$ 
(recall that the separator procedure first triangulates the
input). To overcome this technical issue we consider $C_i$, the external face
of $H_i$, instead of just $S_i$.

For every cycle $C_j$ that is contained in $H_i$ we compute a maximum
preflow to $C_j$ in $H_i$.
We then set $G_i$ to be the part of $G_{i-1}$ not strictly enclosed
by $S_i$ and continue to the next iteration.

When all iterations are done, for every node $v$ of every cycle $C_j$
we compute a maximum $v$-to-$t$ flow in $G_0$. 
However, we do not push more than $\sigma(v)$, the excess flow present at $v$. 
We call this step limited max-flow (Line~\ref{limited}).

\subsection{Correctness of Algorithm~\ref{alg:multisource}}

\begin{definition}(Admissible path)
A $u$-to-$v$ path $P$ is called admissible if $\sigma(u) > 0$ and if
$P$ is residual.
\end{definition}

\begin{lemma}\label{lem:A-to-C'}
Fix an iteration $i$ of the while loop. At any time in that iteration after Line~\ref{push1} is executed
for some cycle $C'$, there are no admissible to-$C'$ paths is $H_i$.
\end{lemma}
\begin{proof}
By induction on the number of iterations of the loop in Line~\ref{first}. For the base case,
immediately after Line~\ref{push1} is executed for cycle $C'$, by
maximality of the preflow pushed when the edges of $C'$ had infinite
capacity, the lemma holds. 
Assume the lemma holds before Line~\ref{push1}
is executed for cycle $C''$ and let $f$ be the flow pushed in that
execution.
Assume for contradiction that after the execution there
exists an admissible $u$-to-$C'$ path $P$ in $H_i$ for some node $u \in H_i$.

If $P$ was residual before $f$ is pushed, then $\sigma(u)$ must have
been zero at that time. Since $P$ is admissible after he push,
$\sigma(u)>0$ after $f$ is pushed. Therefore, 
before the execution, there must have been an admissible $x$-to-$u$
path $R$ in $H_i$ for
some $x \in H_i$. Thus, 
$R \circ P$ is an admissible $x$-to-$C'$ path in $H_i$ before the execution, a contradiction. 

If $P$ was not residual before $f$ was pushed, there must be 
some dart of $P$ whose reverse is used by $f$.
Let $d$ be the latest such dart in $P$. The fact that $\rev(d)$ is
assigned positive flow by $f$ implies that before $f$ is pushed there
exists an admissible path $Q$ from some node $x \in H_i$ to
$\head(d)$, see Fig.~\ref{fig:A-to-C'}. 
By choice of $d$ this implies that $Q \circ P[\head(d),v]$ is an
admissible $x$-to-$C'$ path in $H_i$ before Line~\ref{push1} is executed for  
cycle $C''$, a contradiction.
\begin{figure}
\centerline{\includegraphics[scale=0.4]{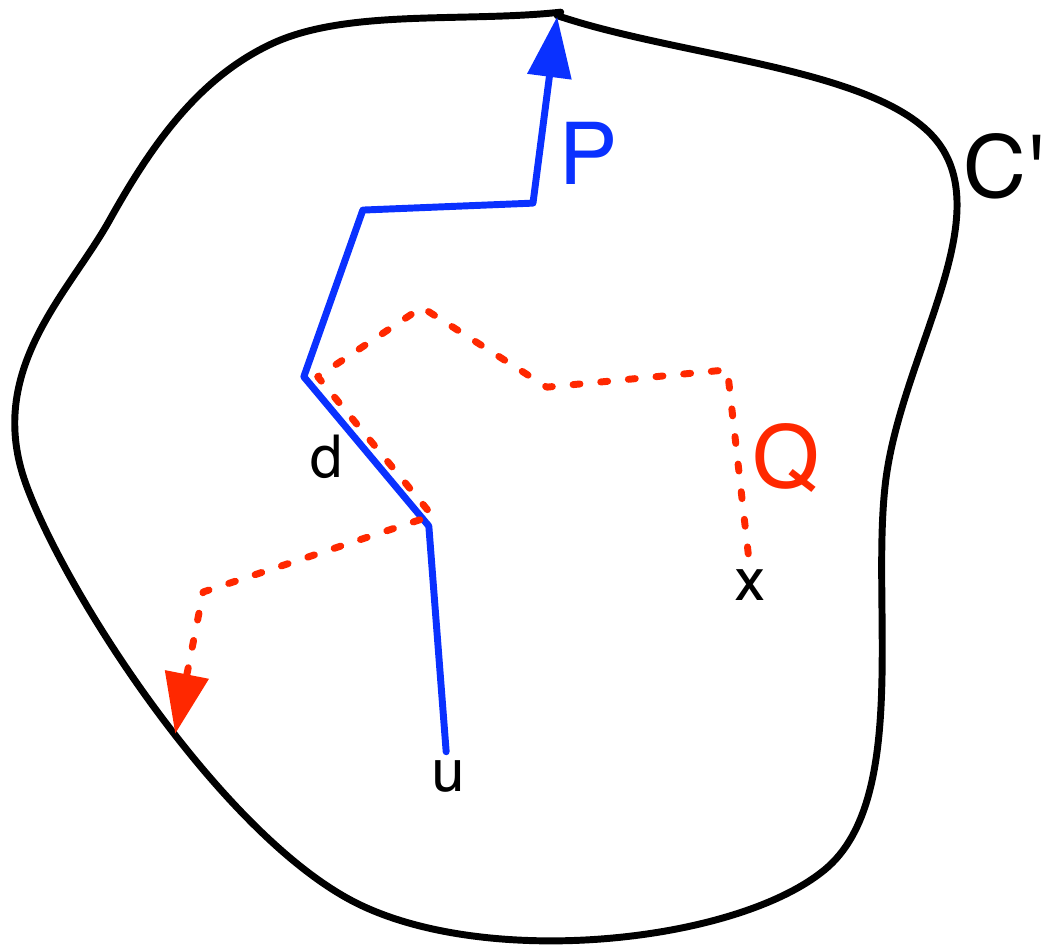}}
\caption{A possible situation in the proof of Lemma~\ref{lem:A-to-C'}.
$P$ is shown in solid blue, $Q$ in dashed red.}
\label{fig:A-to-C'}
\end{figure}
\end{proof}

\begin{lemma}\label{lem:maxpreflow}
Just before the loop in Line~\ref{second} is executed, for every $i$,
there are no admissible to-$B_i$ paths in $H_i$.
\end{lemma}
\begin{proof}
Lemma~\ref{lem:A-to-C'} implies that, for every $i$,  at the end of iteration $i$ of
the while loop, there are no admissible to $B_i$ paths in $H_i$.

Since for $j > i$ $H_i \cap H_j \subseteq C_i $ and since $C_i \in
B_i$, there are no admissible to-$B_i$ paths in $H_i$ at any later iteration as
well.
\end{proof}

\begin{lemma}\label{lem:A-to-t1}
Just before the loop in Line~\ref{second} is executed for the first
time, there are no $v$-to-$t$ admissible paths for any node $v \in G_0 -
\bigcup_j C_j$.
\end{lemma}
\begin{proof}
Let $v$ be a node of $G_0$ that does not belong to any $C_j$. Let $i$ be the
unique index such that $v \in H_i$.
Observe that any $v$-to-$t$  admissible flow path in $G_0$ must visit some node of
$B_i$ before getting to $t$, so it consists of a $v$-to-$C'$ admissible
path in $H_i$ for some $C' \in B_i$, contradicting Lemma~\ref{lem:maxpreflow}
\end{proof}

\begin{lemma}\label{lem:limited}
For any node $u$, if there are no $u$-to-$t$ admissible paths 
before an execution of 
Line~\ref{limited} then there are none after the execution as well.
\end{lemma}
\begin{proof}
If $\sigma(u) = 0$ before the execution
then $\sigma(u) = 0$ after the execution as well, so there are no
admissible $u$-to-$t$ paths.

Otherwise, there is no $u$-to-$T$ residual path before the execution. 
let $f$ be the $v$-to-$t$ flow pushed at Line~\ref{limited}.
Assume for contradiction that after the execution there exists a
$u$-to-$t$ residual path $P$.
Since $P$ was not residual before $f$ was pushed there must be
some dart of $P$ whose reverse is used by $f$.
Let $d$ be the earliest such dart in $P$. The fact that $\rev(d)$ is
assigned positive flow by $f$ implies that before $f$ is pushed there
exists a residual path $Q$ from $\tail(d)$ to $t$.
By choice of $d$ this implies that $P[\tail(d)] \circ Q$ is a
residual (and therefore admissible) $u$-to-$t$ path before the execution, a contradiction.
\end{proof}


We can now prove the correctness of the algorithm claimed in Theorem~\ref{thm:main}.
\begin{proof}(Of correctness of algorithm in Theorem~\ref{thm:main})
By Lemma~\ref{lem:A-to-t1},  immediately before Line~\ref{second}, 
there are no $v$-to-$t$ admissible paths in $G_0$ for any node $v \in G_0 -
\bigcup_j C_j$.
By Lemma~\ref{lem:limited} there are no such paths
after the loop in Line~\ref{second} terminates.
Since the executions of Line~\ref{limited} eliminate all $v$-to-$t$
admissible paths for $v \in \bigcup_j C_j$, there are no admissible
paths to $t$ in $G_0$ upon termination, so
the flow computed is a maximum preflow.
\end{proof}

\subsection{Running Time of Algorithm~\ref{alg:multisource}}

\begin{lemma}\label{lem:2cases}
Every cycle $C_j$ appears at most twice as the cycle stored by the
variable $C$ in 
the loop in line~\ref{first}.
\end{lemma}
\begin{proof}
Consider $C_j$. It appears as the cycle stored by the variable $C$ in
the following two cases: 
\begin{enumerate}
\item when $i=j$ (i.e., when $C_j$ is the external face of $H_i$).
\item when $C_j$ is contained by some $H_i$ but is not the external
 face of $H_i$.
\end{enumerate}
Note that case (2) implies that $C_j$ has some dart that is strictly enclosed by $S_i$,
so this can happen for exactly one value of $i$ since the subgraph
strictly enclosed by $S_i$ is not part of $G_{i+1}$. Thus, $C_j$ is not
contained by any $H_i'$ with $i'>i$.
\end{proof}

We first consider the cost of the recursive calls in Line~\ref{push1}.
In both cases in Lemma~\ref{lem:2cases}, the recursive call is on the graph $H_i$, so if $T(n)$
denotes the running time of Algorithm~\ref{alg:multisource} on an
input graph with $n$ nodes, the cycle $C_j$ contributes at most
$T(|H_j|) + T(|H_{p(j)}|)$, where $p(j)$ is the unique value such that
$C_j$ is contained by $H_{p(j)}$ but is not the external face of
$H_{p(j)}$.
Therefore, the total time required by all recursive calls is 
$$\sum_j T(|H_j|) + T(|H_{p(j)}|).$$

Observe that, since for every $i$, $\frac{1}{3}|G_{i-1}| \leq
|H_i| - 2\sqrt2|G_{i-1}| \leq
\frac{2}{3}|G_{i-1}|$, we have $|H_i| < |H_j|$ for $i>j$.
Also note that, if $C_j$ is not the external face of $H_i$ and $H_i$
contains $C_j$, then $i>j$.  Therefore, $\sum_j T(|H_j|) + T(|H_{p(j)}|)$ is
bounded by $\sum_j T(|H_j|) + T(|H_{j+1}|) = T(|H_1|) + 2 \sum_{j\geq 2} T(|H_j|)$.

\begin{lemma}\label{lem:rec1} 
$\sum_{j=k_1}^{k_2}|H_j|^{1.5} \leq 0.7|G_{k_1-1}|^{1.5}$.
\end{lemma}
\begin{proof}
By induction on $k_2 - k_1$. Recall that, for every $j$, $|H_j| \leq
\frac{2}{3} |G_{j-1}| + 2 \sqrt{2|G_{j-1}|}$. By inspection, for $k_2 =
k_1$, 
$ ( \frac{2}{3} |G_{k_1-1}| + 2 \sqrt{2|G_{k_1-1}|} ) ^{1.5}
\leq 0.7|G_{k_1-1}|$ provided $|G_{k_1-1}| > N_0 = 10^5$.

Assume the claim holds for $k_2 - k_1 = k-1$.
\begin{eqnarray*}
\sum_{j=k_1}^{k_1+k}|H_j|^{1.5} & \leq & 
|H_{k_1}|^{1.5} + \sum_{j=k_1+1}^{k_1+1+k-1}|H_j|^{1.5} \\
& \leq & |H_{k_1}|^{1.5} + 0.7|G_{k_1}|^{1.5} \\
& \leq & \left(  \theta|G_{k_1-1}| + 2\sqrt{2|G_{k_1-1}|}
\right)^{1.5} + 0.7 \left(  (1-\theta)|G_{k_1-1}| + 2\sqrt{2|G_{k_1-1}|}
\right)^{1.5},  
\end{eqnarray*}
where $\theta$ is the balance parameter of the separator $S_{k_1}$.
In the first inequality we have used the inductive assumption.
Using the convexity of the above expression, it can be bounded by
setting $\theta=\frac{2}{3}$, which satisfies the lemma provided that
all graphs have at least $N_0=10^5$ nodes. 
\end{proof}

\begin{lemma}\label{lem:rec}
Assume $T(n) \leq \alpha_1 n^{1.5} \log n$ for every $ N_0 \leq n < |G_0|$. Then,
$T(|H_1|) + 2\sum_{j \geq 2} T(|H_j|) < 0.98 \alpha_1 |G_0|^{1.5} \log |G_0|$ 
\end{lemma}

\begin{proof}

Let $\theta_1$ denote the balance parameter for separator $S_1$. 
\begin{eqnarray*}
T(|H_1|) + 2\sum_{j \geq 2} T(|H_j|) & \leq & \alpha_1
|H_1|^{1.5}\log|G_0| + 2 \alpha_1 \log|G_0| \sum_{j \geq 2}
|H_j|^{1.5} \\
& \leq & \alpha_1
|H_1|^{1.5}\log|G_0| + 2 \cdot 0.7 \alpha_1 \log|G_0| |G_1|^{1.5} \\
& \leq & \alpha_1
(\theta_1 |G_0| + 2 \sqrt{2|G_0|}) ^{1.5}\log|G_0| + \\
& & 2 \cdot 0.7
\alpha_1 \log|G_0| \left( (1-\theta_1)|G_0| + 2\sqrt{2|G_0|} \right)^{1.5},
\end{eqnarray*}
where in the second inequality we have used Lemma~\ref{lem:rec1}.
Using the convexity of the above expression, it can be bounded by
setting $\theta=\frac{1}{3}$, which yields the desired bound $0.98 \alpha_1
|G_0|^{1.5} \log |G_0|$ provided $|G_0| \geq N_0$. 
\end{proof}

\begin{lemma}\label{lem:nodes-of-Ci}
for every $i$, every node of $C_i$ belongs to some Jordan separator $S_j$.
\end{lemma}
\begin{proof}
$C_i$, the external face of $H_i$ consists of nodes that either belong to
$S_i$ or to a face of $G_{i-1}$ that is not triangulated. To see that,
consider a clockwise traversal of $S_i$. For every two consecutive
nodes of $S_i$ that are connected in $G_{i-1}$ by an edge, $S_i$
coincides with that edge (see Section~\ref{sec:separator}), so it belongs to $C_i$. The only parts of
$C_i$ and $S_i$ that do not coincide correspond to $S_i$ crossing
some non-triangulated face $f$ of $G_{i-1}$, say at nodes $u$ and $v$. In those cases, $C_i$ 
consists of the clockwise subpath $f$ between $u$ and $v$.
Since $G_0$ is triangulated in the first line of the algorithm and
Line~\ref{Gi},  
every face of $G_{i-1}$ that is not triangulated
corresponds to regions that were strictly enclosed by previous separators, or
more formally, to a union of the portions of $G_0$ that are strictly
enclosed by some Jordan separators in $\{S_j: j<i$\}. Therefore, every
node on these faces belongs to some Jordan separator $S_j$, which
proves the lemma.
\end{proof}


We can now put together the pieces to prove the running time stated in
Theorem~\ref{thm:main}.

\begin{proof}(Of running time in Theorem~\ref{thm:main}. )
We have already argued that the time required for all recursive calls
is bounded by $\sum_j T(|H_j|) + T(|H_{p(j)}|)$. 
The work done outside the recursive calls is dominated by the
single-source single-sink flow computations in
Line~\ref{limited}. Each of these computations takes $O(|G_0| \log
|G_0|)$ time. 
The overall time required for the non-recursive calls is thus 
$O \left( \left| \bigcup_j C_j \right||G_0| \log |G_0| \right)$. By lemma~\ref{lem:nodes-of-Ci} this is 
$O (\sum_j |S_j||G_0| \log |G_0|)$.
Since $|S_i|$ is $O(\sqrt{|G_{i-1}|})$ and since the size of the
$G_i$s decreases exponentially, we have
\begin{eqnarray*} 
O \left( \sum_j |S_j||G_0| \log |G_0| \right) & \leq & 
\alpha_2 |G_0| \log |G_0| \sum_j \sqrt{|G_{j-1}|} \\
& \leq &  \alpha_2' |G_0|^{1.5} \log |G_0|
\end{eqnarray*}
for some constants $\alpha_2, \alpha_2'$.

The overall running time is therefore bounded by 
$$T(|G_0|) \leq \sum_j T(|H_j|) + T(|H_{p(j)}|) + \alpha_2' |G_0|^{1.5} \log
|G_0|.$$
Assume inductively $T(n) \leq \alpha_1 |G_0|^{1.5} \log |G_0|$ for some
constant $\alpha_1$. Then, by Lemma~\ref{lem:rec}, $T(|G_0| \leq 0.98 \alpha_1
|G_0|^ {1.5} \log |G_0| + \alpha_2' |G_0|^{1.5} \log |G_0|$ which at most
$\alpha_1 |G_0|^{1.5} \log |G_0|$ for appropriate choice of $\alpha_1$.
\end{proof}

\section{Converting a maximum feasible preflow into a maximum flow}\label{sec:preflow}
In this section we describe a linear time algorithm that, given a
feasible preflow in a planar graph, converts it into a feasible flow
of the same value.  This algorithm can be used to convert the maximum
preflow output by Algorithm~\ref{alg:multisource} into a maximum flow.
This section contains no novel ideas and is included for
completeness. A similar procedure was used in \cite{JV82}, but was not
described in detail.

First, use the technique of Kaplan and Nussbaum~\cite{KaplanNussbaum2009} to make the preflow
acyclic. The running time of this step is dominated by a shortest paths computation in the dual of
the residual graph. This can be done in
$O(n \log n)$ using Dijkstra, or in linear time using~\cite{HKRS97}.

Let $f$ denote the acyclic feasible maximum preflow in $G$.
Let $p(v)$ denote the net inflow of node $v$.
Let $D$ denote the DAG induced by arcs with $f(d)>0$. 
Reverse every arc of
$D$ and compute a topological
order on the nodes of $D$. The following algorithm pushes back flow
from nodes with positive net inflow to the sources and runs in linear time. Upon
termination, $f$ is a feasible maximum flow. 
\begin{algorithm}\caption{An algorithm that converts acyclic preflow $p$ on a DAG $D$ into a flow.}
\begin{algorithmic}[1]
\FOR {$v \in D$ in topological order}
   \IF{$v$ is not a sink}
      \WHILE{$p(v) > 0$}
         \STATE let $uv$ be a dart, where $u$ comes after $v$ in
         topological order and $f(d) >0$
         \STATE $x := \min\{f(d),p(v)\}$
         \STATE $f(d) := f(d)-x$
         \STATE $p(v) := p(v)-x$
      \ENDWHILE
  \ENDIF
\ENDFOR
\end{algorithmic}
\end{algorithm}

\section*{Acknowledgments}
We thank Glencora Borradaile for pointing out that computing a maximum preflow
may be useful in solving the multiple source flow problem.

\bibliographystyle{plain}
\bibliography{long,all}

\end{document}